\documentclass[11pt]{article}

\usepackage[english]{babel}
\usepackage[margin=1in]{geometry}

\usepackage{amsmath,amsfonts,amsthm,mathtools}
\usepackage{thmtools,mathrsfs}
\usepackage{xcolor}
\usepackage{booktabs}
\usepackage{tikz}
\usepackage{tabularx}
\usepackage[hidelinks]{hyperref}

\newcolumntype{Y}{>{\raggedright\arraybackslash}X}

\newcommand{\simplex}[1]{\Lambda_{#1-1}}
\newcommand{\rad}{r}

\newtheorem{theorem}{Theorem}[section]
\newtheorem{lemma}[theorem]{Lemma}
\newtheorem{corollary}[theorem]{Corollary}

\theoremstyle{definition}
\newtheorem{definition}{Definition}[section]

\theoremstyle{remark}
\newtheorem{remark}{Remark}[section]

\theoremstyle{plain}
\newtheorem{proposition}{Proposition}[section]

  {\par\medskip}

\newcommand{\vecs}{\mathbf{s}}
\newcommand{\vecu}{\mathbf{u}}

\newcommand{\vecdelta}{\boldsymbol{\delta}}

\title{\textbf{Robustness of Stable Matchings When Attributes and Salience Determine Preferences}}

\author{
Amit Ronen\thanks{Dept. of Computer Science, Bar-Ilan University; {\tt amit.ronen@biu.ac.il}} 
\and 
S.~S. Ravi\thanks{Biocomplexity Institute, University of Virginia and Department of Computer Science, University at Albany--State University of New York; {\tt ssravi0@gmail.com}}
\and 
Sarit Kraus\thanks{Dept. of Computer Science, Bar-Ilan University; {\tt sarit@cs.biu.ac.il}}
}

\date{} 

\begin{document}

\maketitle

\begin{abstract}
In many matching markets--such as athlete recruitment or academic 
admissions--participants on one side are evaluated by attribute vectors known to the other side, 
which in turn applies individual \emph{salience vectors} to assign relative importance to these attributes.  
Since saliences are known to change in practice, a central question arises: 
how robust is a stable matching to such perturbations?
We address several fundamental questions in this context.

First, we formalize robustness as a radius within which a stable matching remains immune to blocking pairs under any admissible perturbation of salience vectors (which are assumed to be normalized). 
Given a stable matching and a radius,
we present a polynomial-time algorithm to verify whether the matching is stable within the specified radius.
We also give a polynomial-time algorithm for computing the maximum robustness radius of a
given stable matching.
Further, we design an anytime search algorithm that uses certified lower and upper bounds 
to approximate the most robust stable matching, and we characterize the robustness-cost relationship through efficiently computable bounds that delineate the achievable tradeoff between robustness and cost.
Finally, we show that for each stable matching, the set of salience profiles that preserve its stability factors is a product of low-dimensional polytopes within the simplex. This geometric structure precisely characterizes the polyhedral shape of each robustness region; its volume can then be computed efficiently, with approximate methods available as the dimension grows, thereby linking robustness analysis in matching markets with classical tools from convex geometry.

\end{abstract}

\section{Introduction}\label{sec:intro}

\subsection{Background and Motivation}
\label{subsec:background}
Stable matching is a foundational model in market design~\cite{Roth1982}.
It underpins real-world allocation systems such as college admissions~\cite{GaleShapley62},
medical residency programs~\cite{RothPeranson99},
school choice~\cite{AbdulkadirogluSonmez03}, roommate allocation \cite{aziz2013stable}, student-project-resource-allocation \cite{ismaili2019student},
and kidney exchange~\cite{RothSonmezUnver04,liu2014internally}.
In the classic bipartite matching problem, two sets of agents must be paired, with each agent having preferences over potential partners on the other side. A matching is considered stable if no pair of agents would both prefer each other to their assigned partners. Traditionally, the model assumes that agents’ preferences are fixed,
complete, and independent~\cite{GusfieldIrving89}.

However, evidence from behavioral science challenges these assumptions \cite{TverskySimonson93}. Human preferences are not arbitrary; they are often dynamic, evolving with context, experience, and time. This raises an important question: if preferences change, can a previously stable matching dissolve (i.e., become unstable)?

Prior research has shown that stable matchings, when preferences are independent of one another, are surprisingly fragile to changes in preferences. Even when preferences are altered in very small ways, stable matchings that remain stable despite such perturbations - that is, matchings that are robust to preference changes - turn out to be exceedingly rare \cite{Chen2021matchings,Boehmer-etal-2024}. 

However, as noted above, agents’ preferences are not arbitrary. In this paper, we focus on settings where agents on one side evaluate potential counterparts through salience-weighted evaluations, where each agent assigns \emph{salience weights} to the observable attributes of potential counterparts~\cite{TeoSethuramanTan06}, reflecting their perceived importance. Crucially, the salience of attributes, and thus the agents’ preferences, may shift over time ~\cite{bordalo2012salience,BordaloGennaioliShleifer13}.

For instance, in college admissions, applicants are characterized by attributes such as 
Scholastic Aptitude Test (SAT) scores and GPA, which are known to colleges. Each college assigns varying levels of importance to these attributes and ranks applicants accordingly.
This motivates a shift in perspective. Rather than thinking about robustness in terms of small swaps in a ranked list, we propose viewing it through the lens of attribute salience. Robustness then becomes a geometric question: how far can these weights shift, within some tolerance, while still preserving stability? In this sense, robustness can be viewed as a \emph{second-order stability} property - a form 
of~ ``stability of stability'' that captures whether stable outcomes persist under changes in preferences.




\noindent
{\bf Model overview.}
We study balanced two-sided markets where side $A$ has static strict preferences and side $B$ induces strict lists by \emph{salience-weighted scores}. Each agent $a \in A$ has a vector $\vecu(a)$ with $m$ non-negative attribute values, for a constant number $m$ of attributes.
Let $\simplex{m}$ denote the $(m-1)$-dimensional probability simplex consisting of all the normalized vectors with $m$ non-negative components.
Each agent $b \in B$ has a salience vector $\vecs(b) \in \simplex{m}$, and
$b$ ranks each $a \in A$ by the scalar product
$\vecs(b)\!\cdot\!\vecu(a)$ along with a given tie-breaking rule.

The model is \emph{hybrid}: continuous in salience space yet ordinal in realized lists, letting us bring geometric tools to classical structures of stable matchings (e.g., deferred acceptance, the lattice of stable matchings, and rotations \cite{GaleShapley62,GusfieldIrving89}). 
The set of all stable matchings forms a distributive lattice ordered by the preferences of one side, ranging from the $B$-optimal to the $A$-optimal outcome. 
Within this lattice, \emph{rotations} represent minimal cyclic exchanges of partners that move monotonically from one stable matching to another, and serve as the fundamental units in our robustness analysis. 
Robustness in this framework is quantified geometrically by a radius~$r$ in an $\ell_p$ norm, capturing the maximum perturbation under which stability persists. 
In practice, decision makers rarely alter all \emph{salience weights} simultaneously, motivating a \emph{support budget}~$k$ that limits how many coordinates of~$\vecs(b)$ may change at once.

\subsection{Summary of Contributions}
\label{subsec:contrib}

While previous work on the robustness of
stable matchings (e.g., \cite{Chen2021matchings}) considered swapping entries in preference lists, our work considers robustness when preference lists are modified by perturbing salience vectors.
We study some fundamental
questions under this notion of robustness.
Our primary contributions are listed below. 

\smallskip 

\noindent
\underline{\textsf{1. A new notion of robustness using salience.}} 
When a salience vector
$\vecs(b)$ of an agent $b \in B$ is perturbed to $\mathbf{\hat{s}}(b)$, 
the distance $r$ between $\vecs(b)$ and 
$\mathbf{\hat{s}}(b)$ under some $\ell_p$ metric is  the \emph{radius}
of the perturbation.
Such perturbations could change the preference list of $b$ and may cause a  matching to become unstable.
We define a stable matching $\mu$ to be 
$(k,r,p)$-\emph{robust} if it remains stable when at most $k \leq m$ components of the salience vector of an agent are modified and the radius of this perturbation under the $\ell_p$
metric is at most $r$.
We examine several algorithmic and geometric problems based on this notion of robustness.

\smallskip 

\noindent
\underline{\textsf{2. Algorithm for robustness verification.}}~ 
Given a stable matching $\mu$, a non-negative real number $r$, a non-negative integer $k \leq m$ and
a value $p \in \{1, 2, \infty\}$,
we present an algorithm to verify whether
$\mu$ is $(k,r,p)$-robust.
For \emph{fixed} $m$, our algorithm
runs in polynomial time as the verification problem reduces to solving a polynomial number of convex programs (linear program (LP) or second-order cone program (SOCP) feasibility checks), each solvable in polynomial time (see e.g., \cite{BoydVandenberghe04,AlizadehGoldfarb03}).

\smallskip 

\noindent
\underline{\textsf{3. Algorithm for finding the maximum robustness 
radius.}}~
Extending the ideas used to solve the verification problem, we show that for any given a stable matching $\mu$,
the maximum radius $r^*(\mu)$ such that $\mu$ is $(k, r^*(\mu), p)$-robust can be computed in polynomial time for fixed $m$
and any $p \in \{1, 2, \infty\}$.

\smallskip 

\noindent
\underline{\textsf{4. Finding the most robust stable matching.}}~ Given the feature vectors of the agents in $A$ and the salience vectors of agents in $B$, we show that lower and upper bounds on the maximum robustness radius $r^*$ of a most robust stable matching can be computed efficiently.
Using these bounds, we present an anytime search algorithm that finds a 
stable matching whose robustness radius is within the computed bounds.
This search algorithm relies on several concepts associated with stable matchings (e.g., deferred acceptance, rotation poset;
see~\cite{GusfieldIrving89,Manlove13}
or Appendix~\ref{app:DA-rotations} for definitions of these concepts).

\smallskip 

\noindent
\underline{\textsf{5. Robustness-cost tradeoffs.}}~ We also examine an extension of robustness that incorporates a cost function capturing welfare, fairness, or other priorities relevant to applications. We provide a polynomial-time algorithm for finding a stable matching whose robustness radius exceeds a given threshold while minimizing the cost among all matchings meeting that robustness requirement.

\smallskip 

\noindent
\underline{\textsf{6. Describing the robustness region of a given matching.}}~
We use the term \emph{salience profile} to denote the matrix $S\in\mathbb{R}^{n\times m}$, where the $b$-th row is the salience vector $\vecs(b)^\top$. Given a matching~$\mu$, the \emph{robustness region} of~$\mu$ consists of all salience profiles~$S$ under which~$\mu$ remains stable. We show that this robustness region factors as a product of \(n\) low-dimensional polytopes within the simplex, fully characterizing its geometry. The appendix further derives how the volume of this region can be computed in polynomial time.

\smallskip 

A summary of the computational complexity of all main problems
appears in Table~\ref{tab:complexity} in the Appendix.

\subsection{Related Work} \label{subsec:relat-wo}

\noindent
\textbf{Foundations and structure.}
We rely on the classical lattice and rotation frameworks of stable matchings~\cite{GusfieldIrving89,Manlove13}, which provide the structural basis for our analysis. 
Polyhedral characterizations, in particular Rothblum’s stable-marriage polytope~\cite{Rothblum94}, provide the geometric foundation for our optimization formulations of robustness and radius computation.

\noindent
\textbf{Robustness under ordinal perturbations.}
A major line of research explores robustness under ordinal perturbations. Chen et al.~\cite{Chen2021matchings} formalize $d$-robustness via Kendall–$\tau$ distance (adjacent swaps) and develop polynomial algorithms to \emph{find} and \emph{optimize} $d$-robust matchings using the rotation partially ordered set (poset), with hardness emerging in the presence of ties. 
Mai and Vazirani~\cite{MaiVazirani2018robust} extend this perspective by considering robustness to \emph{uncertainty in the preference lists themselves}, developing algorithms that preserve stability under bounded perturbations of linear orders. 
In sharp contrast, Boehmer et al.~\cite{Boehmer-etal-2024} show that on large random instances, a stable matching typically fails to survive even a single adjacent swap.  
Earlier work introduced the notion of \emph{supermatches}~\cite{IJCAI2017Supermatch},
defining robustness in terms of repairability: a matching is robust if small disruptions
(i.e., a limited number of pair breakups) can be fixed by a nearby stable matching.
More recent rotation-based presolve algorithms~\cite{IJCAI2024RobustRotation}
use the rotation poset to efficiently identify matchings with maximal robustness
under such ordinal formulations. In contrast, our work introduces a continuous, attribute-based notion of robustness,
where stability is preserved against perturbations in agents’ salience weights
rather than discrete swaps in their ordinal preference lists.

\noindent
\textbf{Behavioral motivation.}
Beyond purely ordinal models, prior work links matching behavior to context-dependent \emph{salience} of attributes in behavioral decision theory~\cite{Bhatnagar-2008}. Similarly, many real-world matching systems implicitly rely on attribute-based evaluations.  
For example, in New York City’s centralized high-school match~\cite{NYCMatch2005}, programs assess applicants using transparent, multi-criteria rubrics (e.g., test performance, attendance, or neighborhood priority), and different schools emphasize these features to varying degrees -- effectively corresponding to distinct salience vectors.

\noindent
\textbf{Preference evolution and dynamic settings.}
Another line of work studies how preference profiles evolve over time.  
Echenique et al.~\cite{Echenique2024experimental} provide experimental evidence on decentralized matching processes with endogenously changing choices.  
Other recent models, such as Bredereck et al.~\cite{Bredereck2019adapting} 
and Alimudin and Ishida~\cite{Alimudin2022updating}, explore algorithmic adaptation of stable matchings under dynamically changing preference profiles.  
Similar adaptive-weight dynamics also appear in multi-agent learning,
where agents adjust feature weights over fixed attribute spaces to
adapt to changing environments or interaction patterns
(e.g.,~\cite{Muslimani2025R2N,JimenezRomero2025SwarmLLM}). 
Our framework complements these approaches by capturing preference change as a \emph{structured geometric perturbation} in salience space, bridging dynamic evolution and robust stability within a unified model.

\noindent
{\bf Roadmap} \label{subsec:roadmap}
We begin by defining the attribute-salience model and our notation (Section~\ref{sec:model}).  
We then study robustness for a given matching (Section~\ref{sec:verify-given}) and show how to compute its maximal robustness radius (Section~\ref{sec:radius-given}).  
Next, we develop algorithms for finding a stable matching that attains the maximal robustness radius (Section~\ref{sec:bound-radius}) and derive bounds on the relationship between robustness and cost (Section~\ref{sec:por}).  
Finally, we explore the geometric structure of robustness regions (Section~\ref{sec:geometry}).

\section{Model and Notation}\label{sec:model}
We define the attribute-salience matching setup (agents, preferences, attributes of $A$, salience of $B$), robustness notions, and notational conventions used throughout. For concreteness, a complete numerical example illustrating all the definitions is provided in Appendix~\ref{app:example}.

\noindent
\textbf{Agents and matchings.}
Let
$A=\{a_1,\dots,a_n\}$ and
$B=\{b_1,\dots,b_n\}$, with
$|A|=|B|=n$,
be two disjoint sets of agents. A \emph{matching} is a bijection $\mu:A\to B$; write
$\mu^{-1}(b)$ for $b$'s partner in $A$.

\noindent
\textbf{Side $A$: static preferences.}
Each agent $a\in A$ has a permanent, strict, complete order over $B$ that
never changes. We use the notation $b\succ_a b'$ 
to indicate that $a$ prefers $b$ to $b'$.

\noindent
\textbf{Attributes of $A$.} For each agent $a \in A$, there is a vector of attributes $\vecu(a)$ that characterizes $a$'s properties. These attributes are observable to agents on side $B$. The vector
$\vecu(a)$ consists of $m\ge 2$ observable attributes. Throughout this paper, we assume that $m$ is a \emph{fixed constant.} In our
notation, $\vecu(a)=(u_1(a),\dots,u_m(a))\in\mathbb{R}_{\geq 0}^m$. 

\smallskip 

\noindent
\textbf{Side $B$: salience vectors.}
Each agent $b\in B$ evaluates candidates by assigning weights to the observable attributes of $A$.  
Formally, we represent these priorities by a salience vector
\[
\vecs(b) = (s_1(b),\dots,s_m(b))\in\simplex{m}:=\{\vecs(b)\in\mathbb{R}_{\ge 0}^m:\sum_i s_i(b)=1\}.
\]

We denote the \emph{salience profile} of side~$B$ by the $n \times m$ matrix $S=(\vecs(b))_{b\in B}\in(\simplex{m})^n$, which collects the salience vectors of all agents in~$B$.
Agent $b \in B$ assigns to each $a\in A$ the score
$
\vecs(b)\!\cdot \vecu(a)=\sum_{i=1}^m s_i(b)\,u_i(a).
$
Agents in $B$ then rank candidates by decreasing score, with public strict tie-breaking orders $\prec^{\mathrm{tie}}$. We write $a\succ_b a'$ when $b$ prefers $a$ over $a'$.

\noindent
\textbf{Ordinal ranks of $A$ and $B$.}
For each agent $a\in A$, let $\operatorname{rank}_a(b)\in\{1,\dots,n\}$ denote the position of $b$ in $a$'s strict preference list (with a lower value indicating a more preferred agent).
For each agent $b\in B$, let $\operatorname{rank}_b(a)\in\{1,\dots,n\}$ denote the position of $a$ in $b$'s strict list induced by the salience rule (sorting $A$ by decreasing 
$\vecs(b)\!\cdot \vecu(\cdot)$ and breaking ties by $\prec^{\mathrm{tie}}$).
Thus $b\succ_a b'$ iff $\operatorname{rank}_a(b)<\operatorname{rank}_a(b')$, and $a\succ_b a'$ iff $\operatorname{rank}_b(a)<\operatorname{rank}_b(a')$.

\noindent
\textbf{Blocking pairs and stability.}
A pair $(a,b)\in A\times B$ is a \emph{blocking pair} for a matching $\mu$ if 
$
b\succ_a \mu(a)
\quad\text{and}\quad
a\succ_b \mu^{-1}(b).
$
A matching $\mu$ is \emph{stable} when no blocking pair exists.

\noindent
\textbf{Radius parameters and worst-case perturbations.}
Fix a support budget $k\in\{1,\dots,m\}$, a radius $\rad\ge 0$, and a norm $p\in\{1,2,\infty\}$. 
We allow a single agent $b \in B$ to \emph{modify} up to $k$ components of its salience vector. Formally, $b$ selects additive perturbations $\boldsymbol{\delta}\in\mathbb{R}^m$ with support 
$Q\subseteq[m]$, $|Q|\le k$, satisfying $s_i(b)+\delta_i\ge 0$ for all $i\in Q$, and defines
\[
T:=\sum_{i=1}^m \bigl(s_i(b)+\delta_i\bigr)>0.
\]
The perturbed vector is then normalized as
\[
\hat{\mathbf{s}}(b)=\frac{\mathbf{s}(b)+\boldsymbol{\delta}}{T}\in\simplex{m}.
\]
The perturbation radius $r$ is computed post-normalization,
\[
\|\hat{\mathbf{s}}(b)-\mathbf{s}(b)\|_p \le \rad.
\]

\noindent
A matching~$\mu$ is said to be \emph{$(k,\rad,p)$-robust} if it remains stable under all perturbations of the above form satisfying this bound for every agent~$b\in B$.
Note that $\mu$ is $(k,0,p)$-robust iff $\mu$ is stable at the salience profile $S$.

\noindent
\textbf{Pre- vs.\ post-normalization.}
An equivalent description replaces $T$ by its reciprocal $\lambda=1/T>0$. 
In this \emph{post-normalized} view, perturbations are specified by a vector 
$\hat{\mathbf{s}}(b)\in\simplex{m}$ and a scalar $\lambda>0$ such that 
\[
\hat s_i(b)=\lambda\, s_i(b)\quad (i\notin Q), \qquad 
\|\hat{\mathbf{s}}(b)-\mathbf{s}(b)\|_p \le \rad.
\]
Thus pre- and post-normalization are simply two parameterizations of the same admissible set, and they yield the same robustness radius.
Throughout the paper, we adopt the post-normalized form for clarity and consistency.

\noindent
\textbf{Score margin.}
To quantify how strongly $b$ prefers its current partner over another candidate $a$, 
we define the \emph{attribute-gap vector}
\[
\boldsymbol{\Delta}(b;a\mid\mu):=\vecu(\mu^{-1}(b))-\vecu(a).
\]
Given a salience profile $S$, the corresponding \emph{score margin} is
\[
\gamma_S(b;a\mid\mu)\ :=\ \vecs(b)\cdot\boldsymbol{\Delta}(b;a\mid\mu).
\]
The value $\gamma_S(b;a\mid\mu)$ is positive when $b$ prefers its current partner
$\mu^{-1}(b)$ to $a$ under profile $S$.
When $\mu$ and/or $S$ are clear from the context, we write simply $\gamma(b;a)$.

\noindent
\textbf{Base radius.}
The \emph{base inner radius} $\rad^{\mathrm{base}}(\mu)$ 
is a baseline robustness guarantee:
it denotes a perturbation level that is sufficient to keep~$\mu$ stable.
By construction $\rad^{\mathrm{base}}(\mu)\le\rad^*(\mu)$,
and it depends only on score margins and dual-norm attribute gaps,
independent of the support budget~$k$
(see Section~\ref{subsec:rot-base}).

\noindent
\textbf{Stability region.}\label{par:Pmu-notation}
For a given matching $\mu$, the \emph{stability region} $\mathcal{P}_\mu$ 
is the set of all salience profiles under which $\mu$ is stable. 
Formally, $\mathcal{P}_\mu\subseteq(\simplex{m})^n$; see Section~\ref{sec:geometry} for the explicit polyhedral form and product structure of 
$\mathcal{P}_\mu$.

\noindent
\textbf{Remark ($k=m-1$ vs.\ $k=m$).}
Since unchanged coordinates automatically adjust to preserve normalization, 
varying $m-1$ coordinates already spans all feasible perturbations.  
A higher support size $k=m$ matters only in boundary cases where the support itself changes 
(e.g., $m=2$, $(1,0)\!\to\!(0,1)$).

\noindent
For convenience, Table~\ref{tab:notation} in the Appendix provides a complete
notation summary of all symbols used throughout the paper.

\section{Verification for a Given Matching}\label{sec:verify-given}

Given a salience profile $S=\{\vecs(b)\}_{b\in B}$ and a matching $\mu$, the verification problem asks whether $\mu$ remains stable under 
\emph{every} admissible perturbation of the salience vector of a single agent $b\in B$ under the worst-case perturbation model.
The main result of this section is that the verification problem can be solved in polynomial time.

For $b\in B$, let $\mathcal{H}_\mu(b):=\{\,a\in A:\ b\succ_a \mu(a)\,\}$ be the set of $A$-agents who prefer $b$ to their current partner. We can now formalize the verification task more precisely by defining the robustness verification problem, which captures the requirement that no new blocking pair can emerge under bounded perturbations of the salience profile.

\begin{definition}[Robustness Verification (RV)]
Given a stable matching $\mu$, attributes $\{\vecu(a)\}_{a\in A}$, salience profile $S$, norm $p\in\{1,2,\infty\}$, support budget $k\le m$, and radius $r\ge 0$, 
decide if, for all $b\in B$, $a\in\mathcal{H}_\mu(b)$, and post-normalized perturbations $\hat{\vecs}(b)\in\simplex{m}$, 
the score margin satisfies the condition
$\gamma_{\hat S}(b;a)\ \ge 0$.
\end{definition}

Robustness verification ensures that no blocking pair can arise 
under any admissible perturbation of the salience vectors.
The next lemma expresses this condition as a local check:
for each~$b$ and each possible perturbation direction,
stability holds iff all resulting score margins remain nonnegative.

\begin{lemma}\label{lem:verify-poly}
For any integers $m$ and $k$ with
$k\le m$, a value $p\in\{1,2,\infty\}$, 
and a rational value $r\ge 0$, 
a stable matching $\mu$ is $(k,r,p)$-robust if and only if, 
for every $b\in B$, $a\in\mathcal{H}_\mu(b)$, and support set $Q\subseteq[m]$ with $|Q|\le k$, 
no admissible perturbation $\hat{\vecs}(b)$ yields a negative score margin, i.e.,
$\gamma_{\hat S}(b;a)<0$.
\end{lemma}

\begin{proof}
\vspace{-2pt}

(\emph{Only if}.) Assume $\mu$ is $(k,r,p)$-robust. 
Assume otherwise that there exist $b\in B$, $a\in\mathcal{H}_\mu(b)$, 
and an admissible perturbation $\hat{\vecs}(b)$ with $\|\hat{\vecs}(b)-\vecs(b)\|_p\le r$ 
such that $\gamma_{\hat S}(b;a)<0$. 
Then $(a,b)$ becomes a blocking pair after the perturbation, contradicting robustness. 
Hence no such perturbation exists.

(\emph{If}.) Conversely, assume that for all $b\in B$, $a\in\mathcal{H}_\mu(b)$, 
and supports $Q$ with $|Q|\le k$, every admissible perturbation $\hat{\vecs}(b)$ 
satisfies $\gamma_{\hat S}(b;a)\ge 0$. 
Thus each $b$ still weakly prefers its partner over any $a\in\mathcal{H}_\mu(b)$, 
and strict tie-breaking ensures no new blocking pair can arise. 
Therefore $\mu$ is $(k,r,p)$-robust.
\end{proof}

Having reduced the condition to finitely many checks, we now establish that the overall verification task admits a polynomial-time algorithm.
\begin{theorem}[Polynomial-time verification via support enumeration]\label{thm:verify-poly}
Let $m$ and $k\le m$ be fixed constants, $p\in\{1,2,\infty\}$, and $r\ge 0$. 
Verifying whether a given stable matching $\mu$ is $(k,r,p)$-robust can be done in polynomial time.
\end{theorem}

\begin{proof} 
\vspace{-7pt}
The key idea is to translate each potential blocking deviation into a convex feasibility problem that tests whether a perturbation within distance~$r$ can violate stability. In particular, using Lemma~\ref{lem:verify-poly}, we reduce RV to testing the infeasibility of $O(n^2 m^k)$ convex optimization instances: a linear program (LP) when $p\in\{1,\infty\}$, or a second-order cone program (SOCP) when $p=2$. Since $m$ and $k$ are constants and the infeasibility of convex instances can be determined in polynomial time~\cite{BoydVandenberghe04,AlizadehGoldfarb03}, 
We obtain a polynomial-time algorithm for RV.

\smallskip
\noindent\emph{Explicit LP/SOCP formulations.}
For each $b\in B$, $a\in\mathcal{H}_\mu(b)$, and support $Q\subseteq[m]$ with $|Q|\le k$, 
we test whether there exists an admissible perturbation 
$\hat{\vecs}(b)\in\simplex{m}$ satisfying 
$\|\hat{\vecs}(b)-\vecs(b)\|_p\le r$
that makes $(a,b)$ a blocking pair, i.e.\ violates stability. Each feasibility check is a convex program - an LP when $p\in\{1,\infty\}$ or an SOCP when $p=2$ (see Appendix~\ref{app:SOCP} for the general SOCP form).
For illustration, the LP for $p=\infty$ is shown below; the remaining formulations appear in Appendix~\ref{app:verify-formulations}.

\smallskip
\noindent\emph{(a) $p=\infty$ (box distance) - LP feasibility for $(a,b,Q)$}
\[
\begin{aligned}
\text{find }&\hat{\vecs}(b),\ \lambda>0\\
\text{s.t. }&
\sum_i \hat s_i(b)=1,\quad \hat s_i(b)\ge0,\\
&\hat s_i(b)=\lambda\, s_i(b)\ (i\notin Q),\\
&\hat{\vecs}(b)\cdot\boldsymbol{\Delta}(b;a\mid\mu)\le0,\\
&-r\le \hat s_i(b)-s_i(b)\le r\quad(\forall i).
\end{aligned}
\]
\emph{Variables:} $\hat{\vecs}(b)\in\mathbb{R}^m_{\ge0}$, $\lambda\in\mathbb{R}_{>0}$.

\smallskip

Collectively, these feasibility programs cover all potential deviations across agents and supports.
Each instance tests whether a perturbation of magnitude at most~$r$ exists
that makes $(a,b)$ a blocking pair.  
If all are infeasible, $\mu$ is $(k,r,p)$-robust; otherwise, a feasible instance yields a blocking witness.  
Since LP and SOCP infeasibility can be decided in polynomial time 
by interior-point methods~\cite{AlizadehGoldfarb03,BoydVandenberghe04}, 
and since $m$ and $k$ are constants, the overall runtime is polynomial in~$n$. \qedhere
\end{proof}

\noindent When $k=m$ (full-support moves), support enumeration vanishes and the procedure reduces to $O(n^2)$ independent LP/SOCP checks indexed by $(a,b)$. 

\begin{remark}[Strict vs.\ non-strict inequalities]\label{rem:strict}
We use non-strict inequalities ($\le$) in the LP/SOCP formulations for solver compatibility. 
Strict tie-breaking ensures that this is equivalent to strict stability, 
since every unperturbed score margin $\gamma(b;a\mid\mu)$ is positive.
\end{remark}

\section{Maximum Radius of a  Matching}\label{sec:radius-given}
Having established how to verify robustness for a given radius~$\rad$, 
we now turn to the complementary task: 
determining the largest radius for which stability still holds.  
In other words, for a given stable matching~$\mu$, support size~$k$, and norm~$p$, 
we ask: \emph{what is the maximum perturbation radius under which~$\mu$ remains robust?}  
This value quantifies the exact tolerance of~$\mu$ to worst-case changes and serves as a natural definition of robustness.
Intuitively, the robustness radius is determined by the weakest link: 
the smallest perturbation, over all potentially blocking pairs $(a,b)$ and admissible supports $Q$, that makes $b$ prefer $a$ over its current partner $\mu^{-1}(b)$.
\begin{definition}[Pairwise thresholds and robustness radius]\label{def:pairwise-thresholds}
For each $(a,b,Q)$, let $\rad^{\min}(b;a\mid Q)$ denote the optimal value of the corresponding LP 
(for $p\in\{1,\infty\}$) or SOCP (for $p=2$). Aggregating over all admissible supports, define $\rad^{\min}(b;a)$ by
\[
\rad^{\min}(b;a) := \min_{\substack{Q\subseteq [m]\\ |Q|\le k}} \rad^{\min}(b;a\mid Q).
\]
The \emph{maximum robustness radius} of $\mu$ is then
\[
\rad^*(\mu) := \sup\{\,r\ge0:\ \mu \ \text{is $(k,r,p)$-robust}\,\}
= \min_{\substack{b\in B \\ a\in\mathcal{H}_\mu(b)}} \rad^{\min}(b;a).
\]
\end{definition}
The following lemma provides the basis for our polynomial-time algorithm for computing the maximum radius.
\begin{lemma}\label{lem:rmin-convex-prog}
For each fixed $m$ and $k$ with $k\le m$, any norm $p\in\{1,2,\infty\}$, 
and a stable matching $\mu$, each term $\rad^{\min}(b;a)$ from
Definition~\ref{def:pairwise-thresholds}
can be computed in polynomial time.
\end{lemma}
\begin{proof}
For a given $a$, $b$, and support $Q\subseteq[m]$, 
$\rad^{\min}(b;a\mid Q)$ is obtained by minimizing the radius~$r$
over admissible perturbations of~$\vecs(b)$ supported on~$Q$ 
that make $(a,b)$ a blocking pair.  
Formally, each $\rad^{\min}(b;a\mid Q)$ is the optimal value of a convex optimization instance, that is:
a linear program (LP) when $p\in\{1,\infty\}$, 
or a second-order cone program (SOCP) when $p=2$.

The explicit LP/SOCP formulations for all cases are provided in Appendix~\ref{app:radius-formulations}.  
Their construction is identical to the verification instances 
of Section~\ref{sec:verify-given}, 
except for the following modifications:
\begin{enumerate}
    \item the radius $r$ is now a decision variable,
    \item the objective is to minimize~$r$, and
    \item the additional constraint $r\ge0$ is included.
\end{enumerate}

Since the feasibility of each convex optimization instance is decidable in polynomial time and there are $O(m^k)$ supports per pair $(a,b)$ for given $k<m$, 
each $\rad^{\min}(b;a)$ can be computed in polynomial time. 
When $k=m$, support enumeration collapses to a single program per $(a,b)$, so the complexity reduces to one convex program.
\end{proof}
Building on Lemma~\ref{lem:rmin-convex-prog}, we can now establish the main result of this section.
\begin{theorem}[Polynomial-time computation via pairwise thresholds]\label{thm:radius-poly}
Let $m$ and $k\le m$ be fixed constants, and let $p\in\{1,2,\infty\}$. 
For any stable matching $\mu$, the robustness radius $\rad^*(\mu)$  is computable in polynomial time.
\end{theorem}
\begin{proof}
\vspace{-5pt}

For any stable matching $\mu$, the robustness radius $\rad^*(\mu)$ is given by 
Definition~\ref{def:pairwise-thresholds}. As shown above, for each pair $(a,b)$,
the value $\rad^{\min}(b;a)$ can be computed by solving $O(m^k)$ 
convex optimization instances. 
Since there are $O(n^2)$ pairs
$\rad^*(\mu)$ can be found by solving  $O(n^2 m^k)$ convex  instances.
Since $m$ and $k$ are fixed, and each convex instance can be solved in polynomial time, the theorem follows.
As explained earlier, the number of convex instances 
reduces to $O(n^2)$ when $k=m$.    
\end{proof}

{\sloppy \section{Finding the Most Robust Matching: Bounds and Anytime Search}}\label{sec:bound-radius}
\noindent In the previous sections, we assumed that a matching had already been selected and examined its robustness properties.
However, it is often more advantageous to identify in advance the most robust stable matching.
Therefore, our next goal is to \emph{select the stable matching with the largest robustness radius}, namely
\[
\mu^\star \;\in\; \arg\max_{\mu\in\mathcal{SM}}\ \rad^*(\mu),
\]
where $\mathcal{SM}$ denotes the set of all stable matchings of the market and $\rad^*(\mu)$ is the exact robustness radius from Section~\ref{sec:radius-given}.
A challenge arises because the size of $\mathcal{SM}$ can be exponential in~$n$. A classical way to explore this space is via the \emph{rotation poset}~\cite{GusfieldIrving89,Manlove13}, which compactly represents exponentially many stable matchings and supports polynomial-time traversal and optimization under structural criteria (see Section~\ref{subsec:rot-base} for details).
Unfortunately, $\rad^*(\mu)$ itself does not align with the rotation-poset structure: perturbations that create blocking pairs need not correspond to a single rotation or any set of rotations. To address this misalignment, we introduce the \emph{base inner radius} $\rad^{\mathrm{base}}(\mu)$, a conservative proxy satisfying $\rad^{\mathrm{base}}(\mu)< \rad^*(\mu)$.  
This proxy can be expressed in closed form from local score margins and dual gaps, and its constraints align with the rotation poset, enabling efficient search for highly robust matchings.  
While $\rad^*(\mu)$ can be computed by solving $O(n^2 m^k)$ convex optimization instances, $\rad^{\mathrm{base}}$ can be found in $O(n^2 m)$ time.

In this section, Section~\ref{subsec:rot-base} recalls the classical rotation-poset machinery and defines $\rad^{\mathrm{base}}$, together with its structural and computational properties. In sections~\ref{subsec:LB-muB}-\ref{subsec:anytime}, we tackle the main problem of finding the most robust stable matching, combining lower/upper bounds with an anytime search.

\subsection{Rotations and the base inner radius}\label{subsec:rot-base}

\textbf{Classical background.}
We recall standard notions from stable matching; full details appear in Appendix~\ref{app:DA-rotations} (see also \cite{GusfieldIrving89,Manlove13}).  
$B$-proposing DA (deferred acceptance) returns the $B$-optimal matching $\mu_B$.  
The set of stable matchings forms a distributive lattice with $\mu_A \preceq \mu \preceq \mu_B$.  
A \emph{rotation $\rho$ exposed at $\mu$} is a cyclic sequence of pairs $((a_1,b_1),\dots,(a_\nu,b_\nu))$ whose elimination yields another stable matching $\mu'=\operatorname{elim}(\mu,\rho)$. 
Along a rotation, the $B$-side weakly improves and the $A$-side weakly worsens.  
Rotations admit a partial order; every stable $\mu$ can be written as $\mu=\operatorname{elim}(\mu_A,D)$ for a unique down-set $D$; there are $O(n^2)$ rotations in total, and the rotation poset can be built in $O(n^2)$ time.

\noindent
\textbf{Deterministic inner radius from score margins.}
Set $p\in\{1,2,\infty\}$ and let $p^\star$ denote its dual norm, defined by $1/p+1/p^\star=1$.
Recall the score margin $\gamma(b;a\mid\mu)$ from Section~\ref{sec:model}, 
which captures how strongly $b$ prefers its partner over~$a$.  
Together with the partner-dependent dual gap
\[
U_{p^\star}(b)\ :=\ \max_{a'\neq \mu^{-1}(b)} 
   \|\vecu(\mu^{-1}(b))-\vecu(a')\|_{p^\star},
\]
we obtain the \emph{base inner radius}:
\begin{equation}\label{eq:base-radius}
\rad^{\mathrm{base}}(\mu)\ :=\ (1-\varepsilon_{\mathrm{base}})
\min_{b\in B}\ \min_{a:\,\mu^{-1}(b)\succ_b a}\ 
\frac{\gamma(b;a)}{U_{p^\star}(b)}.
\end{equation}

Here $\varepsilon_{\mathrm{base}}\in(0,1)$ is a small fixed constant 
used only to ensure strict feasibility of inequalities.  
We also assume that the attribute vectors $\{\vecu(a):a\in A\}$ are not all identical in their attributes, 
so that $U_{p^\star}(b)>0$ for every $b\in B$.

\begin{lemma}[Margin $\Rightarrow$ radius]\label{lem:margin-radius}
If $\rad\le\rad^{\mathrm{base}}(\mu)$, then $\mu$ remains stable under all
$\ell_p$ perturbations of radius $\rad$.
\end{lemma}

\begin{proof}
\vspace{-8pt}

For each $b$ and $a\neq\mu^{-1}(b)$, write
\[
\hat{\vecs}(b)\cdot(\vecu(\mu^{-1}(b))-\vecu(a))
=\gamma(b;a)+(\hat{\vecs}(b)-\vecs(b))\cdot\boldsymbol{\Delta}(b;a\mid\mu).
\]
By Hölder’s inequality~\cite[App.~A.1.6]{BoydVandenberghe04}), $|(\hat{\vecs}(b)-\vecs(b))\cdot\boldsymbol{\Delta}(b;a\mid\mu)|
\le \|\hat{\vecs}(b)-\vecs(b)\|_p\,\|\boldsymbol{\Delta}(b;a\mid\mu)\|_{p^\star}
\le \rad\,U_{p^\star}(b)$.  
Hence the perturbation term is at least $-\rad U_{p^\star}(b)$.  
Thus the residual margin is $\ge \gamma(b;a)-\rad U_{p^\star}(b)$.  
If $\rad\le\rad^{\mathrm{base}}(\mu)$ this remains positive, hence stability is preserved
\end{proof}

\smallskip
\noindent
As an immediate consequence, we obtain the following corollary.

\begin{corollary}[Base inner radius]\label{cor:base-inner}
If $\rad\le \rad^{\mathrm{base}}(\mu)$, then $\mu$ remains stable 
under all perturbations $\hat{\vecs}(b)$ with 
$\|\hat{\vecs}(b)-\vecs(b)\|_p \le \rad$ for every $b\in B$.  
Equivalently, $\rad^*(\mu)> \rad^{\mathrm{base}}(\mu)$.
\end{corollary}

\subsection{A polynomial-time computable lower bound via $B$-optimal matching}\label{subsec:LB-muB}

To search for the most robust stable matching, it is essential to start from a
guaranteed baseline. A natural choice is the $B$-optimal matching $\mu_B$,
the outcome of $B$-proposing deferred acceptance. By construction, each
$b\in B$ receives its most preferred attainable partner among all stable matchings (a classical property of $\mu_B$)~\cite{Manlove13}. This maximizes the score
margins $\gamma(b;a)$ against less-preferred candidates, and hence yields
the largest base inner radius $\rad^{\mathrm{base}}(\mu)$ across the lattice
of stable matchings.

Since $\mu_B$ is itself a stable matching, its robustness radius directly 
yields a lower bound:
\begin{equation}
\mathbf{LB} := \rad^*(\mu_B)\ \le\ \max_{\mu\in\mathcal{SM}} \rad^*(\mu).
\end{equation}

\begin{proposition}[Polynomial-time LB]\label{prop:LB-muB}
The lower bound $\mathbf{LB}=\rad^*(\mu_B)$ is computable in polynomial
time by solving $O\!\bigl(n^2\binom{m}{k}\bigr)$ LP/SOCP instances
from Section~\ref{sec:radius-given}, which reduces to $O(n^2)$ when $k=m$.
\end{proposition}

\noindent
This gives a tractable starting point: although $\mu_B$ need not maximize robustness, its radius is efficiently computable and provides a guaranteed
baseline against which all further improvements can be evaluated.

\subsection{A relaxation-based global upper bound}\label{subsec:UB-global}

To upper-bound $\max_\mu \rad^*(\mu)$ we relax integrality, working with 
the \emph{stable-marriage polytope} $\mathrm{conv}(\mathcal{SM})$ \cite{Rothblum94}.  
This polytope is the convex hull of incidence matrices 
$X^\mu\in\{0,1\}^{A\times B}$ that encode stable matchings $\mu$, 
where $X^\mu_{ab}=1$ if $a$ is matched to $b$ under $\mu$ and $0$ otherwise.

Although the convex hull may have exponentially many vertices (since the number of stable matchings can be exponential in $n$), the stable-matching polytope admits a compact linear description of polynomial size~\cite{Rothblum94}. Hence, all computations over this convex relaxation remain tractable. Thus, instead of a single integral matching, we allow fractional convex combinations of stable matchings, which yields a tractable convex relaxation.

For each $b\in B$, define the \emph{fractional partner attribute vector}
\[
\bar{\vecu}_b \;=\; \sum_{a\in A} X_{ab}\,\vecu(a), \qquad X\in\mathrm{conv}(\mathcal{SM}).
\]
Intuitively, $\bar{\vecu}_b$ averages the attribute vectors of $b$'s partners across
the mixture $X$. Robustness of radius $r\ge0$ requires that even under post-normalized
$\ell_p$ perturbations of radius $r$, each $b$ prefers $\bar{\vecu}_b$ to every
alternative $a\in A$.

\noindent
\textbf{Case $k=m$.}
When perturbations may reweight all $m$ coordinates, Hölder’s inequality yields the exact no-blocking constraint
\begin{equation}\label{eq:no-block-relax}
\vecs(b)\!\cdot\!\big(\bar{\vecu}_b-\vecu(a)\big)\ \ge\ 
r\;\Big\|\big(\vecu(a)-\bar{\vecu}_b\big)_+\Big\|_{p^\star},
\qquad \forall a\in A,\ b\in B,
\end{equation}
where $(\cdot)_+$ denotes coordinate-wise positive part. Unlike the absolute value, negative coordinates are set to zero, capturing only attributes where $a$ dominates $\bar{\vecu}_b$.
The right-hand side represents the largest margin loss that $b$ could suffer if a perturbation concentrates all weight on attributes where $a$ is stronger than $\bar{\vecu}_b$.  
Thus, if \eqref{eq:no-block-relax} holds, $b$ cannot be tempted by $a$ under any $\ell_p$ perturbation of radius $r$.

\noindent
\textbf{Case $k<m$.}
When perturbations may alter at most $k$ coordinates of $\vecs(b)$, the worst-case loss is smaller.  
Let $\mathbf{v}=\vecu(a)-\bar{\vecu}_b$ and let $v_{[1]}\ge v_{[2]}\ge\cdots\ge v_{[m]}$ denote the coordinates of $\mathbf{v}_+$ sorted in nonincreasing order.  
The maximum reduction in margin under $\ell_p$ radius $r$ and support budget $k$ is then
$
r\cdot \Big\| (v_{[1]},\dots,v_{[k]}) \Big\|_{p^\star}.
$
Accordingly, the no-blocking constraint becomes
\begin{equation}\label{eq:no-block-relax-k}
\vecs(b)\!\cdot\!\big(\bar{\vecu}_b-\vecu(a)\big)\ \ge\ 
r\;\Big\|\big(\vecu(a)-\bar{\vecu}_b\big)^{(k)}_+\Big\|_{p^\star},
\qquad \forall a\in A,\ b\in B,
\end{equation}
where $(\cdot)^{(k)}_+$ denotes the vector formed by the $k$ largest positive coordinates of $\vecu(a)-\bar{\vecu}_b$ (zeros elsewhere).  
This is the $k$-aware support norm, which interpolates between the unrestricted case ($k=m$) and the degenerate case $k=1$.

\noindent
\textbf{Feasibility and optimization.}
Checking feasibility of \eqref{eq:no-block-relax}--\eqref{eq:no-block-relax-k}
reduces to convex optimization: linear programming for $p \in \{1,\infty\}$ and
second-order cone programming for $p=2$, both solvable in polynomial time.
We therefore define the relaxation-based upper bound as
\begin{equation}\label{eq:UB-def}
\mathbf{UB}\ :=\ 
\max\Bigl\{r\ge0:\ \exists\,X\in\mathrm{conv}(\mathcal{SM})\ \text{satisfying \eqref{eq:no-block-relax} or \eqref{eq:no-block-relax-k}}\Bigr\}.
\end{equation}
\begin{theorem}[Relaxation-based UB]\label{thm:UB}
$\mathbf{UB}$ is computable to additive accuracy $\varepsilon_{\mathrm{UB}}$ in polynomial time via bisection with LP/SOCP feasibility tests. 
Moreover,
$
\max_{\mu\in\mathcal{SM}}\ \rad^*(\mu)\ \le\ \mathbf{UB}.
$
If the optimizer $X^\star$ is integral (i.e., corresponds to some stable matching $\mu^\star$), then $\mathbf{UB}=\rad^*(\mu^\star)$ and the maximizing stable matching $\mu^\star$ is recovered explicitly, certifying exact optimality.
\end{theorem}
\begin{proof}
\vspace{-8pt}

Feasibility for a fixed $r$ is a convex program, an LP for $p\!\in\!\{1,\infty\}$ 
and an SOCP for $p\!=\!2$, and thus solvable in polynomial time.
Bisection on $r$ yields an additive $\varepsilon_{\mathrm{UB}}$-approximation 
in polynomial time.
Since $\mathrm{conv}(\mathcal{SM})$ contains all integral stable matchings,
$\max_{\mu\in\mathcal{SM}}\rad^*(\mu)\le\mathbf{UB}$.
If the optimal $X^\star$ is integral, it corresponds to a stable matching 
$\mu^\star$ achieving equality.
\end{proof}

\noindent
\textbf{Local frontier bounds.} Beyond the global value $\mathbf{UB}$, 
the same relaxation can be restricted to sublattices of $\mathcal{SM}$ 
defined by down-sets $D$ in the rotation poset.  
This yields local bounds $UB(D)$ that refine $\mathbf{UB}$ during the anytime search.  
Formally, $UB(D)$ is obtained by adding the rotation constraints of $D$ 
to the global relaxation, so the feasible region only shrinks.
These refinements are computable in polynomial time without recomputing from scratch, 
and they integrate directly into the anytime search 
(Section~\ref{subsec:anytime}), where they ensure that the global LB/UB 
relation continues to hold.

\subsection{Anytime search on the rotation poset}\label{subsec:anytime}

We now describe an anytime search procedure that combines the 
lower bound from $\mu_B$ and the relaxation-based upper bounds 
to progressively narrow the gap between them.

\paragraph{Algorithm.}  
We maintain a priority queue of nodes $D$ in the rotation poset, 
each corresponding to a downset of rotations and its associated stable matching 
$\mu_D=\operatorname{elim}(\mu_A,D)$. 
The queue is  in descending order of $UB(D)$ values, so that at each step we 
explore the matching whose sublattice still allows the largest possible 
robustness radius.

\begin{enumerate}
\item \textbf{Initialization.}  
Insert the root node $D=\emptyset$ (corresponding to $\mu_A$) 
into the queue, and set $\mathbf{LB}\gets \rad^*(\mu_B)$.

\item \textbf{Node extraction.}  
Remove from the queue the node $D$ with the largest $UB(D)$.

\item \textbf{Pruning.}  
If $UB(D)\le \mathbf{LB}$, discard $D$, since no matching reachable 
from it can improve the current best robustness.

\item \textbf{Exact evaluation.}  
Compute $\rad^*(\mu_D)$ exactly (via the LP/SOCP formulations from 
Section~\ref{sec:radius-given}) and update $\mathbf{LB}$ if improved.

\item \textbf{Expansion.}  
For each rotation $\rho$ exposed at $D$, form $D' = D \cup \{\rho\}$, 
compute $UB(D')$, and insert it into the priority queue. 
\end{enumerate}
The search continues until either $\mathbf{LB}=\mathbf{UB}_{\mathrm{frontier}}$ (certifying exact optimality) or until a predefined expansion budget is reached, after which the best matching found so far is returned.
\begin{theorem}[Anytime correctness]\label{thm:anytime}
At all times,
\[
\mathbf{LB}\ \le\ \max_\mu \rad^*(\mu)\ \le\ \mathbf{UB}_{\mathrm{frontier}}\ \le\ \mathbf{UB}.
\]
$\mathbf{LB}$ increases monotonically while $\mathbf{UB}_{\mathrm{frontier}}$ decreases monotonically.  
If some frontier node attains an integral relaxation, exact optimality is certified:
$
\mathbf{LB}=\mathbf{UB}_{\mathrm{frontier}}=\rad^*(\mu^\star).
$
\end{theorem}
\begin{proof}
\vspace{-5pt}

Each node $D$ corresponds to a downset of rotations and a stable matching $\mu_D$.  
Since $UB(D)$ is derived from a relaxation of the exact feasibility region,
we have $\rad^*(\mu_D)\le UB(D)\le\mathbf{UB}$ for all~$D$.  
$\mathbf{LB}$ records the best exact radius found so far and can only increase, while $\mathbf{UB}_{\mathrm{frontier}}=\max_D UB(D)$ over the active nodes can only decrease as nodes are expanded or pruned.  
Thus the stated inequalities and monotonicity follow.  
If an integral relaxation is reached, it corresponds to a stable matching $\mu^\star$ with
$\mathbf{LB}=\mathbf{UB}_{\mathrm{frontier}}=\rad^*(\mu^\star)$, certifying exact optimality.
\end{proof}

Although this anytime search is not polynomial-time in the worst case, it always maintains certified bounds $[\mathbf{LB},\mathbf{UB}_{\mathrm{frontier}}]$ on the optimum. 
Both bounds evolve monotonically ($\mathbf{LB}$ increases, 
$\mathbf{UB}_{\mathrm{frontier}}$ decreases), so the procedure converges toward the true optimum and can be stopped at any point with provable bounds on the achieved robustness.

\section{Robustness--Cost Tradeoffs}\label{sec:por}

In the previous section, we studied robustness alone, aiming to find a 
stable matching with the largest radius.  
A natural next step is to combine robustness with a \emph{cost function} which can model 
welfare, fairness, or other priorities in applications. Fixing a robustness at a given value $\tau\ge0$, we ask for the minimum cost achievable under this requirement. Formally, for separable costs $C(\mu)=\sum_{a\in A} c_{a,\mu(a)}$ and 
robustness requirement $\tau\ge0$, the target is
\[
C^*(\tau)\;=\;\min\{\,C(\mu):\ \rad^*(\mu)\ge\tau\,\}.
\]
Direct optimization under $\rad^*(\mu)$ is computationally difficult,
since the exact radius does not align with the rotation-poset structure. Instead, we develop polynomial-time proxy quantities that yield certified (i.e., provable) upper and lower bounds on $C^*(\tau)$.

\subsection{Upper bound via the base radius}

The base radius $\rad^{\mathrm{base}}(\mu)$ 
(Section~\ref{subsec:rot-base}) is computable in closed form and always satisfies 
$\rad^{\mathrm{base}}(\mu)\le \rad^*(\mu)$.  
Requiring $\rad^{\mathrm{base}}(\mu)\ge\tau$ is equivalent to insisting that 
each $b\in B$ is matched only to candidates $a$ for which 
$\gamma(b;a)/U_{p^\star}(b)\ge\tau$, where $U_{p^\star}(b)$ is the 
dual-norm attribute gap defined in Section~\ref{subsec:rot-base}.  
In terms of the standard rotation poset, this condition simply prunes 
all stable matchings that violate the base-radius threshold, leaving a 
distributive sublattice over which optimization can be carried out.
\begin{theorem}[Polynomial-time base frontier]\label{thm:base-cost}
For any $\tau\ge0$, the minimum-cost stable matching with 
$\rad^{\mathrm{base}}(\mu)\ge\tau$ is computable in polynomial time.
Hence
\[
C^*(\tau)\ \le\ C^{\mathrm{UB}}(\tau)\ :=\ \min\{\,C(\mu):\ 
\rad^{\mathrm{base}}(\mu)\ge\tau\,\}.
\]
\end{theorem}

\begin{proof}

Fix $\tau \ge 0$.
For any stable matching~$\mu$, the constraint
$\rad^{\mathrm{base}}(\mu) \ge \tau$
requires that every eliminated rotation~$\rho$ satisfy
\[
\min_{(a,b)\in\rho}
\big[
\vecs(b)\!\cdot\!(\bar{\vecu}_b-\vecu(a))
-\tau\,\|(\vecu(a)-\bar{\vecu}_b)_+\|_{p^\star}
\big] \ge 0,
\]
where $(\cdot)_+$ denotes the coordinate-wise positive part.
Hence, the set of rotations that can be eliminated while maintaining
stability at level~$\tau$ forms a \emph{downward-closed subset}
of the rotation poset.

We assign each rotation a modified weight
\[
\Delta_\tau(\rho)
=\min_{(a,b)\in\rho}
\big[
\vecs(b)\!\cdot\!(\bar{\vecu}_b-\vecu(a))
-\tau\,\|(\vecu(a)-\bar{\vecu}_b)_+\|_{p^\star}
\big]
-\!\!\sum_{(a,b)\in\rho} c_{a,b}.
\]
Minimizing $C(\mu)$ subject to $\rad^{\mathrm{base}}(\mu)\!\ge\!\tau$
is therefore equivalent to finding a
\emph{maximum-weight closure}
in the rotation poset.
By the classical reduction of maximum-weight closure
to a single $s$--$t$ min-cut~\cite{GusfieldIrving89},
the optimal stable matching for each fixed~$\tau$
can be computed in polynomial time.

Distinct threshold values of~$\tau$
arise only when some base constraint becomes tight, i.e.,
$
\vecs(b)\!\cdot\!(\bar{\vecu}_b-\vecu(a))
=\tau\,\|(\vecu(a)-\bar{\vecu}_b)_+\|_{p^\star}.
$
Since there are $O(n^2)$ score margins $(a,b)$
and $O(n)$ rotations in the poset,
the number of breakpoints is $O(n^3)$.
Sweeping over these values yields the exact
$(\rad^{\mathrm{base}},C)$ frontier.
\end{proof}
\subsection{Lower bound via relaxation}
To obtain a lower bound, we again work with Rothblum’s 
stable-marriage polytope $\mathrm{conv}(\mathcal{SM})$ 
introduced in Section~\ref{subsec:UB-global}, 
but now augment it with \emph{vulnerability cuts}:  
for each cross pair $(a,b')$ and $(a',b)$ that could block within 
radius $\tau$, we add the constraint $y_{ab'}+y_{a'b}\le1$. Here $y_{ab}\in[0,1]$ denotes the standard assignment variable,
and $\mathrm{conv}(\mathcal{SM})$ admits a polynomial-size linear description~\cite{Rothblum94}. Formally, we solve
\[
\begin{aligned}
\min\ & \sum_{a,b} c_{ab}\,y_{ab} \\
\text{s.t. } \quad & y\in\mathrm{conv}(\mathcal{SM}), \quad y \ \text{satisfies all vulnerability cuts for }\tau.
\end{aligned}
\]
Every $\tau$-robust stable matching $\mu$ induces a feasible point in 
this relaxation, so the LP optimum $C^{\mathrm{LB}}(\tau)$ is a valid 
lower bound.

\begin{proposition}[LP lower bound]\label{prop:LB-cost}
Every stable $\mu$ with $\rad^*(\mu)\ge\tau$ induces a feasible LP solution.  
Therefore
\[
C^*(\tau)\ \ge\ C^{\mathrm{LB}}(\tau).
\]
\end{proposition}

The vulnerable set changes only at $O(n^4)$ thresholds, since each arises from a cross pair $(a,b'),(a',b)$.  
Thus, at most $O(n^4)$ distinct LPs need to be solved.  
Each such LP carries $O(m^k)$ additional constraints from the support budget; 
hence, $C^{\mathrm{LB}}(\tau)$ can be computed by solving a total of $O(n^4 m^k)$ LP instances.

\subsection{Certified tradeoff bounds}

For each robustness target $\tau$, our two constructions yield
\[
C^{\mathrm{LB}}(\tau)\ \le\ C^*(\tau)\ \le\ C^{\mathrm{UB}}(\tau).
\]
Here $C^{\mathrm{UB}}(\tau)$ comes from requiring the base radius 
$\rad^{\mathrm{base}}(\mu)\ge\tau$, which produces an explicit stable 
matching and a constructive upper bound.  
The lower bound $C^{\mathrm{LB}}(\tau)$ arises from the LP relaxation 
with vulnerability cuts, giving a valid numerical bound that coincides with a true stable matching whenever the LP optimum 
is integral. Varying $\tau$ reveals the robustness-cost tradeoff: stricter robustness raises the minimum cost, and the bounds capture this relationship.

\section{Geometry of the Robustness Region}\label{sec:geometry}

Up to this point we examined robustness under a fixed salience profile~$S$.  
When both a matching~$\mu$ and a profile~$S$ are given, 
we asked whether $\mu$ remains stable under a given perturbation radius 
and computed its exact robustness radius.  
When only~$S$ is given, we searched for the most robust matching 
and studied the tradeoff between robustness and cost. We now take the dual view: fix a matching~$\mu$ and characterize, geometrically, its \emph{robustness region} -- the set of all salience profiles under which~$\mu$ remains stable. A quantitative analysis of its volume appears in Appendix~\ref{app:volume}.

While the robustness radius captures the most fragile perturbation direction 
(the minimal deviation that breaks stability),  
the robustness region reveals the full multidimensional structure of stability:
its \emph{geometry} indicates which perturbation directions are tolerated,
and its \emph{volume} (analyzed in Appendix~\ref{app:volume}) quantifies the total range of salience profiles for which $\mu$ remains stable.
This allows us to compare matchings not only by their most fragile direction,
but also by the overall extent of their stability region in salience space.  
This contrasts with Rothblum’s \emph{stable-marriage polytope} -  
a polytope in matching space whose vertices are stable matchings.  
Our robustness region instead lives in the salience space 
and forms the \emph{profile polytope} of~$\mu$.

From Section~\ref{par:Pmu-notation}, the robustness region $\mathcal P_\mu$ of a matching $\mu$ 
is defined by the set of salience profiles under which no blocking pair arises.  
Equivalently, stability reduces to the linear inequalities
\begin{equation}\label{eq:ineq}
\vecs(b)\cdot\bigl(\vecu(\mu^{-1}(b))-\vecu(a)\bigr)\ \ge\ 0
\qquad \forall b\in B,\ a\in\mathcal H_\mu(b),
\end{equation}
where $\mathcal H_\mu(b)$ is the set of candidate blockers as introduced in Section~\ref{sec:verify-given}.

\begin{lemma}[Factorization]\label{lem:factorization}
The robustness region factorizes across $B$:
\[
\mathcal P_\mu\ =\ \prod_{b\in B}\ \mathcal P_\mu(b),
\]
where each $\mathcal P_\mu(b)$ is the polytope
\[
\mathcal P_\mu(b)\ :=\ \{\vecs\in\simplex{m}:\ 
\vecs\cdot(\vecu(\mu^{-1}(b))-\vecu(a))\ge 0,\ 
\forall a\in\mathcal H_\mu(b)\}.
\]
\end{lemma}

\begin{proof}
This lemma follows directly from the fact that each stability constraint
involves only the salience vector~$\vecs(b)$ of a single agent~$b$.
\end{proof}

\begin{proposition}[Polyhedral structure]\label{prop:polytope}
Each factor $\mathcal P_\mu(b)$ is a convex polytope in $\simplex{m}$ defined by at most $|\mathcal H_\mu(b)|\le n-1$ linear inequalities.  
Therefore $\mathcal P_\mu$ is a polytope in $(\simplex{m})^n$ defined by $O(n^2)$ inequalities in total.
\end{proposition}

\begin{proof}
\vspace{-7pt}

Each inequality in \eqref{eq:ineq} is linear in $\vecs(b)$ and together with 
the simplex constraints $\vecs(b)\ge0$, $\sum_i s_i=1$, they define a bounded convex polyhedron.  
Hence $\mathcal P_\mu(b)$ is a polytope, i.e., a bounded intersection of finitely many half-spaces inside the simplex. Summing over all $b$, the total number of inequalities is $O(n^2)$.  
\end{proof}

\noindent
\textbf{Illustration.}  
Figure~\ref{fig:simplex2} shows an example of a $\mathcal P_\mu(b)$ for a single agent $b$ when $m=3$. The simplex $\Lambda_2$ is visualized as an equilateral triangle in the plane, 
where each interior point corresponds to a valid salience vector of $b$.  
In this example, $H_\mu(b)=4$, so four potential blocking agents induce linear indifference constraints, drawn as dashed lines.  
The shaded 2-dimensional polytope (i.e., polygon) is their feasible intersection, consisting of all salience vectors for which $b$ does not deviate from~$\mu$.

\vspace{-6pt}
\begin{figure}[t]
\centering
\begin{tikzpicture}[scale=1.0]

  \coordinate (A) at (0,0);
  \coordinate (B) at (4,0);
  \coordinate (C) at (2,3.464); 

  \draw[thick] (A)--(B)--(C)--cycle;

  \node[below left] at (A) {$u_1$};
  \node[below right] at (B) {$u_2$};
  \node[above] at (C) {$u_3$};


  \coordinate (P1) at (1.3,0.7);
  \coordinate (P2) at (2.7,0.7);
  \coordinate (P3) at (2.4,1.7);
  \coordinate (P4) at (1.6,1.8);

  \filldraw[fill=gray!25,draw=black]
    (P1)--(P2)--(P3)--(P4)--cycle;

\usetikzlibrary{calc}

\draw (P1) -- (P2);
\draw[dashed]  ($(P1)!-0.6!(P2)$) -- ($(P1)!1.6!(P2)$);

\draw (P2) -- (P3);
\draw[dashed]  ($(P2)!-0.7!(P3)$) -- ($(P2)!2.7!(P3)$);

\draw (P3) -- (P4);
\draw[dashed]  ($(P3)!-0.8!(P4)$) -- ($(P3)!1.7!(P4)$);

\draw (P4) -- (P1);
\draw[dashed]  ($(P4)!-1.4!(P1)$) -- ($(P4)!1.7!(P1)$);

  \node at (2.0,1.2) {$\mathcal P_\mu(b)$};

\end{tikzpicture}
\vspace{-8pt}
\caption{Example of $\mathcal P_\mu(b)$ for $m=3$: the shaded polygon inside the simplex $\Lambda_2$.}
\label{fig:simplex2}
\vspace{-10pt}

\end{figure}
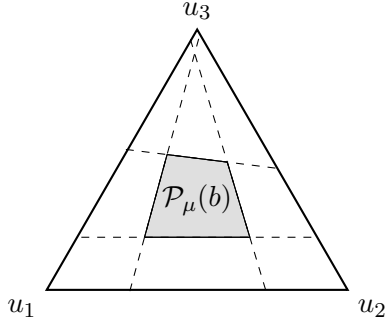

\medskip
Testing whether $S\in\mathcal P_\mu$ is simply the standard stability check 
for~$\mu$ under the salience profile~$S$, a polynomial-time check over all $O(n^2)$ pairs. This observation motivates a geometric analysis of the full stability region.
The complete analysis of the polyhedral structure and volume computation of 
$\mathcal P_\mu$ is provided in Appendix~\ref{app:volume}.
\section{Summary and Conclusions}
\label{sec:concl}

We introduced a new notion of robustness for stable matchings based on attributes and salience vectors.
We presented polynomial-time algorithms for several problems in this context, including verifying whether a given matching remains stable for a specified radius, computing its maximum stability radius, and approximating the most robust stable matching through efficiently computable bounds. We also extended the framework to incorporate costs, deriving computable upper and lower bounds that relate robustness to the cost of stability. Finally, we showed that the robustness region of a given matching factorizes as a product of low-dimensional polytopes within the simplex. Overall, our results establish a unified geometric and algorithmic foundation for analyzing stability under structured preference perturbations and open several promising directions for further study.

\section*{Acknowledgments}
This research has been partially supported by the Israel Science Foundation under grant 2544/24.

\appendix

\clearpage

\bibliographystyle{plain} 
\bibliography{sample}

\clearpage

\onecolumn

\begin{center}
\fbox{\Large\textbf{Supplementary Material}}
\end{center}

\medskip

\noindent
\textbf{Paper title:} Robustness of Stable Matchings When Attributes and Salience Determine Preferences

\smallskip

\noindent
\textbf{Submission ID:} 809

\bigskip

\appendix

\begin{table}[t]
\centering
\small
\caption{Computational complexity of the main problems.}\label{tab:complexity}
\begin{tabularx}{\linewidth}{@{}p{.55\linewidth}Y@{}}
\toprule
\textbf{Problem} & \textbf{Complexity / Result} \\
\midrule
Robustness Verification (RV) & $O(n^2 m^k)$ LP/SOCPs \\
Robustness Radius Computation (RRC) & $O(n^2 m^k)$ LP/SOCPs \\
Most Robust Stable Matching (FMRM) & $O(n^2 m^k)$ lower bound; LP/SOCP upper bound; anytime search \\
Robustness-Cost Tradeoff (RCB) & $O(n^3)$ for base frontier; $O(n^4 m^k)$ for LP lower bound \\
Geometry of Robustness Region (GRR) & Polyhedral characterization (exact volume derivation in appendix) \\
\bottomrule
\end{tabularx}

\medskip
\small
\noindent
\emph{Note.} When $k=m$, the combinatorial factor $m^k$ drops out, and the corresponding complexities simplify to their $O(n^2)$ or $O(n^4)$ forms.
\end{table}

\begin{table}[t]
\centering
\small
\caption{Notation used throughout the paper.}\label{tab:notation}
\begin{tabularx}{\linewidth}{@{}>{\raggedright\arraybackslash}p{.25\linewidth}X@{}}
\toprule
\textbf{Symbol} & \textbf{Meaning} \\
\midrule
$A=\{a_1,\dots,a_n\}$, $B=\{b_1,\dots,b_n\}$ & Two disjoint sets of agents, each of size $n$ \\
$\mu:A\to B$ & A matching (bijection between $A$ and $B$) \\
$\mu^{-1}(b)$ & The partner of $b$ under matching $\mu$ \\
$b\succ_a b'$ & Agent $a$ prefers $b$ to $b'$ (strict order) \\
$a\succ_b a'$ & Agent $b$ prefers $a$ to $a'$ \\
$\operatorname{rank}_a(b)$, $\operatorname{rank}_b(a)$ & Ordinal position of a partner in $a$’s or $b$’s list \\
$m$ & Attribute dimension (constant) \\
$\vecu(a)=(u_1(a),\dots,u_m(a))\in\mathbb{R}_{\ge0}^m$ & Attribute vector of agent $a$ (public) \\
$\simplex{m}$ & $(m{-}1)$-simplex $\{\,s\in\mathbb{R}_{\ge0}^m:\sum_i s_i=1\,\}$ \\
$\vecs(b)=(s_1(b),\dots,s_m(b))\in\simplex{m}$ & Salience vector of agent $b$ \\
$S=(\vecs^1,\dots,\vecs^n)\in(\simplex{m})^n$ & Salience profile of all $B$-side agents \\
$\vecs(b)\!\cdot\!\vecu(a)$ & Score of candidate $a$ under salience vector $\vecs(b)$ \\
$\mathcal{H}_\mu(b)=\{a\in A:\ b\succ_a \mu(a)\}$ & Set of $A$-agents who prefer $b$ to their current partner \\
$(a,b)$ blocking pair & $b\succ_a\mu(a)$ and $a\succ_b\mu^{-1}(b)$ \\
$\|\cdot\|_p$ & $\ell_p$ norm ($p\in\{1,2,\infty\}$) used for perturbation distance \\
$k$ & Support budget (number of salience coordinates allowed to change) \\
$\rad$ & Perturbation radius (tolerance to salience drift) \\
$\hat{\vecs}(b)=\dfrac{\vecs(b)+\boldsymbol{\delta}}{T}$ & Perturbed and renormalized salience vector \\
$T=\sum_i (s_i(b)+\delta_i)$, $\lambda=1/T$ & Pre- and post-normalization parameters \\
$\vecdelta\in\mathbb{R}^m$ & Additive perturbation vector \\
$\boldsymbol{\Delta}(b;a\mid\mu)=\vecu(\mu^{-1}(b))-\vecu(a)$ & Attribute-gap vector for $(a,b)$ under $\mu$ \\
$\gamma_S(b;a\mid\mu)=\vecs(b)\!\cdot\!\boldsymbol{\Delta}(b;a\mid\mu)$ & Score margin of $b$’s partner over candidate $a$ \\
$\rad^{\min}(b;a\mid Q)$ & Minimal perturbation radius that makes $(a,b)$ blocking under support $Q$ \\
$\rad^{\min}(b;a)$ & Minimal radius over all supports $Q$, $\min_{|Q|\le k}\rad^{\min}(b;a\mid Q)$ \\
$\rad^*(\mu)$ & Maximum robustness radius of matching $\mu$ \\
$\rad^{\mathrm{base}}(\mu)$ & Base inner radius (closed-form conservative guarantee) \\
$p^\star$ & Dual norm of $p$ ($1/p+1/p^\star=1$) \\
$U_{p^\star}(b)$ & Dual-norm attribute gap $\max_{a'\neq \mu^{-1}(b)}\|\vecu(\mu^{-1}(b))-\vecu(a')\|_{p^\star}$ \\
$\mathcal{P}_\mu$ & Robustness (stability) region of $\mu$ in salience space \\
$\mathcal{P}_\mu(b)$ & Local stability polytope for $b$ inside $\simplex{m}$ \\
$\mathrm{Vol}(\mathcal{P}_\mu)$ & Exact volume of robustness region $\mathcal{P}_\mu$ \\
$\mathbf{LB},\,\mathbf{UB}$ & Certified lower/upper bounds on the maximal robustness radius \\
$\mathbf{UB}_{\mathrm{frontier}}$ & Current frontier upper bound during the anytime search (tightens monotonically) \\
$C(\mu)=\sum_{a\in A}c_{a,\mu(a)}$ & Separable cost of matching $\mu$ \\
$C^*(\tau)$ & Minimum cost among matchings with robustness $\ge\tau$ \\
$\mathcal{SM}$ & Set of all stable matchings in the market \\
$\operatorname{conv}(\mathcal{SM})$ & Rothblum’s stable-marriage polytope (convex hull of all stable matchings) \\

\bottomrule
\end{tabularx}
\end{table}

\section{Illustrative Running Example: College Admissions}\label{app:example}

We illustrate the model using a small two-by-two college-student market with
two observable attributes: \emph{GPA} and \emph{SAT score} ($m=2$).
For convenience in this example, we normalize each student's attribute vector
so that its coordinates sum to one (this is not required by the model, but
simplifies numerical interpretation).

\paragraph{Agents and attributes.}
There are two students and two colleges:
\[
A=\{a_1,a_2\}, \qquad B=\{b_1,b_2\}.
\]
Each student $a_i$ is represented by an attribute vector
\[
\vecu(a_1)=(0.8,0.2), \qquad \vecu(a_2)=(0.4,0.6),
\]
where the first coordinate is the normalized GPA score and the second is
the normalized SAT score.

\paragraph{Salience (admission priorities).}
Each college~$b$ evaluates students by a salience vector
$\vecs(b)\in \Lambda_1$ representing the relative weight assigned to GPA and SAT:
\[
\vecs(b_1)=(0.7,0.3), \qquad \vecs(b_2)=(0.3,0.7).
\]
Hence college~$b_1$ favors GPA while $b_2$ emphasizes SAT.

\paragraph{Scores and preferences.}
The evaluation scores $\vecs(b)\!\cdot\!\vecu(a)$ are
\[
\begin{array}{c|cc}
 & a_1 & a_2 \\ \hline
b_1 & 0.7(0.8)+0.3(0.2)=0.62 & 0.7(0.4)+0.3(0.6)=0.46 \\[2pt]
b_2 & 0.3(0.8)+0.7(0.2)=0.38 & 0.3(0.4)+0.7(0.6)=0.54
\end{array}
\]
Thus $b_1$ prefers $a_1$ and $b_2$ prefers $a_2$ under the initial salience profile.  

Assume students have static preferences
\[
a_1\!: b_1\succ b_2, 
\qquad 
a_2\!: b_1\succ b_2,
\]
so both students prefer the more GPA-oriented college $b_1$.
Then the unique stable matching is
\[
\mu(a_1)=b_1,\qquad \mu(a_2)=b_2.
\]
Under this configuration, $(a_2,b_1)$ is the most likely blocking pair
once $b_1$ shifts its salience toward the SAT attribute.

\paragraph{Perturbations and robustness.}
Suppose college~$b_1$ shifts its emphasis toward the SAT attribute.
Starting from $\vecs(b_1)=(0.7,0.3)$, consider an additive perturbation
\[
\boldsymbol{\delta}=(-0.2,+0.3), \qquad 
\vecs(b_1)+\boldsymbol{\delta}=(0.5,0.6).
\]
To restore normalization, divide by $T=0.5+0.6=1.1$, giving
\[
\hat{\vecs}(b_1)=\tfrac{1}{1.1}(0.5,0.6)=(0.45,0.55).
\]
The perturbed margin becomes
\[
\hat{\vecs}(b_1)\!\cdot\!\boldsymbol{\Delta}(b_1;a_2\mid\mu)
=0.45(0.4)+0.55(-0.4)=-0.04,
\]
so $(a_2,b_1)$ now forms a blocking pair.

The $\ell_1$ distance between the original and normalized salience vectors is
$\|\hat{\vecs}(b_1)-\vecs(b_1)\|_1=0.5$,
while $\|\hat{\vecs}(b_1)-\vecs(b_1)\|_2\approx0.35$ and $\|\hat{\vecs}(b_1)-\vecs(b_1)\|_\infty=0.25$.
These distances illustrate the robustness radius $\rad^*(\mu)$
under $\ell_1$, $\ell_2$, and $\ell_\infty$ norms, respectively. Here $k=2$ (full support), though in this example the same blocking perturbation could also be achieved with $k=1$.

\paragraph{Support and normalization.}
Here $m=2$, so the perturbation uses the full support $Q=\{1,2\}$.
The normalization factor is $\lambda=1/T\approx0.91$, ensuring
$\sum_i\hat s_i(b_1)=1$ and $\hat{\vecs}(b_1)\in \Lambda_1$.
This post-normalized view corresponds exactly to the form used
in the convex programs of Sections~\ref{sec:verify-given}--\ref{sec:radius-given}.

\paragraph{Geometric intuition.}
Each $\vecs(b)$ lies on the simplex $\Lambda_1$ (the unit line segment).
The stability region $\mathcal{P}_\mu\subseteq(\simplex{2})^2$
collects all salience profiles for which no blocking pair arises.
When $k=m$, the robustness radius $\rad^*(\mu)$
equals the minimum distance from the current profile $S$
to the boundary of~$\mathcal{P}_\mu$ where a blocking pair first appears.

\section{Background: Deferred Acceptance and Rotations}\label{app:DA-rotations}

We work with two equally sized sets $A$ and $B$ ($n=|A|=|B|$); every agent has a \emph{complete, strict} preference list over the other side (no ties).
\paragraph{Matchings and stability (basics).}
A \emph{matching} is a bijection $\mu:A\to B$. For $a\in A$ write $\mu(a)$ for $a$'s partner in $B$, and for $b\in B$ write $\mu^{-1}(b)$ for $b$'s partner in $A$. 
For preferences, we use $b\succ_a b'$ to mean that $a$ prefers $b$ over $b'$, and $a\succ_b a'$ to mean that $b$ prefers $a$ over $a'$. 
A pair $(a,b)\in A\times B$ is a \emph{blocking pair} for $\mu$ if $b\succ_a \mu(a)$ and $a\succ_b \mu^{-1}(b)$. 
A matching $\mu$ is \emph{stable} if it admits no blocking pair.
\smallskip
We write $\succeq_a$ and $\succeq_b$ for the weak orders induced by $\succ_a$ and $\succ_b$,
respectively (i.e., $x\succeq_a y$ means $a$ weakly prefers $x$ to $y$; similarly for $b$).

\subsection{Deferred Acceptance (DA): from first principles}\label{app:DA}
\paragraph{B-proposing DA (one-to-one, strict, complete).}
Initialize all agents unmatched. While some $b\in B$ is unmatched and has not proposed to everyone:
\begin{enumerate}
  \item $b$ proposes to the most-preferred $a\in A$ that $b$ has not yet proposed to.
  \item $a$ tentatively keeps their favorite among all proposers so far \emph{and} their current tentative partner (if any), and rejects all others.
\end{enumerate}
Return the final tentative matches.

\begin{proposition}[Termination and basic invariants]\label{prop:DA-terminate}
B-proposing DA terminates after at most $n^2$ proposals, hence in $O(n^2)$ time. Throughout the run:
(i) each $b$ only moves \emph{down} their list; (ii) each $a$’s tentative partner is never worse (on $a$’s list) than at any earlier time; (iii) once rejected, a $b$ is never reconsidered by that $a$.
\end{proposition}

\begin{theorem}[Stability and optimality]\label{thm:DA-opt}
The output of B-proposing DA is stable (no blocking pairs) and is the $B$-optimal stable matching $\mu_B$ (every $b$ weakly prefers $\mu_B$ to any other stable matching). Symmetrically, A-proposing DA returns the $A$-optimal matching $\mu_A$.
\end{theorem}

DA is “propose-hold-reject’’: proposers never climb back up their lists; acceptors only trade up. These two monotonicities force termination and rule out blocking pairs in the limit.

\paragraph{Lattice of stable matchings.}
For matchings $\mu,\mu'$, write $\mu \preceq \mu'$ iff every $b\in B$ weakly prefers
its partner in $\mu'$ to its partner in $\mu$ (i.e., $\mu'(b)\succeq_b \mu(b)$).
With this partial order, the set of stable matchings forms a distributive lattice with
bottom $\mu_A$ and top $\mu_B$; for any stable $\mu$,
\[
\mu_A \ \preceq\ \mu \ \preceq\ \mu_B.
\]

\subsection{Rotations and the rotation poset}\label{app:rotations}
\begin{definition}[Rotation]\label{def:rotation}
Let $\mu$ be a stable matching. A \emph{rotation} exposed at $\mu$ (i.e., currently eliminable at $\mu$) is a cyclic sequence
\[
\rho=\big((a_1,b_1),(a_2,b_2),\dots,(a_\nu,b_\nu)\big),
\]
such that the following hold:
(i) $\mu(a_i)=b_i$ for all $i=1,\dots,\nu$;
(ii) for $i=1,\dots,\nu-1$, $a_{i+1}$ is the most-preferred agent \emph{above} $a_i$ on $b_i$’s list who would accept $b_i$ if $b_i$ left $a_i$, and in addition $a_1$ is the most-preferred agent \emph{above} $a_\nu$ on $b_\nu$’s list who would accept $b_\nu$ if $b_\nu$ left $a_\nu$;
(iii) for $i=1,\dots,\nu-1$, $a_{i+1}$ prefers $b_i$ to $b_{i+1}=\mu(a_{i+1})$, and in addition $a_1$ prefers $b_\nu$ to $b_1=\mu(a_1)$.
\end{definition}

Figure~\ref{fig:rotation-cycle} illustrates a rotation $\rho=((a_1,b_1),\dots,(a_\nu,b_\nu))$ exposed at $\mu$: gray edges show the current pairs $(a_i,b_i)$, and eliminating $\rho$ reassigns $b_i$ to $a_{i+1}$ for $i=1,\dots,\nu-1$ and $b_\nu$ to $a_1$, yielding the stable matching $\mu'=\operatorname{elim}(\mu,\rho)$. By Proposition~\ref{prop:elim}, every $b\in B$ weakly improves (strictly for $b\in\{b_1,\dots,b_\nu\}$) and every $a\in A$ weakly worsens (strictly for $a\in\{a_1,\dots,a_\nu\}$), with all other agents unchanged.

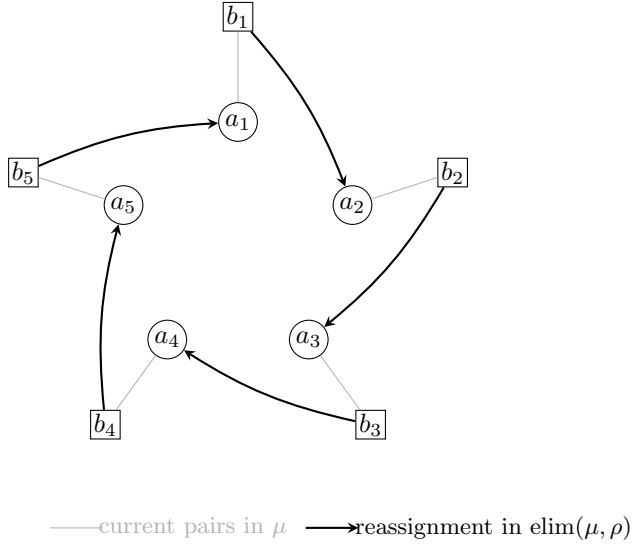
\begin{figure}[t]
\centering
\begin{tikzpicture}[>=stealth, every node/.style={font=\small}]
  \def\n{5}
  \def\rA{1.6}
  \def\rB{3.0}

  \foreach \i in {1,...,\n} {
    \node[circle,draw,inner sep=1.2pt] (a\i) at ({90-360/\n*(\i-1)}:\rA) {$a_{\i}$};
    \node[rectangle,draw,inner sep=1.2pt] (b\i) at ({90-360/\n*(\i-1)}:\rB) {$b_{\i}$};
    \draw[gray!60] (a\i) -- (b\i);
  }

  \foreach \i/\j in {1/2,2/3,3/4,4/5,5/1} {
    \draw[->,thick] (b\i) to[bend left=10] (a\j);
  }

  \node[gray!60] at (-0.6,-3.8) {\footnotesize current pairs in $\mu$};
  \draw[gray!60] (-1.8,-3.8) -- (-2.5,-3.8);
  \draw[->,thick] (0.9,-3.8) -- (1.6,-3.8);
  \node at (3.4,-3.8) {\footnotesize reassignment in $\operatorname{elim}(\mu,\rho)$};
\end{tikzpicture}
\caption{Rotation schematic (exposed at $\mu$).
Gray edges show the current pairs $(a_i,b_i)$; arrows show the reassignment in $\operatorname{elim}(\mu,\rho)$.}
\label{fig:rotation-cycle}
\end{figure}

\begin{proposition}[Elimination]\label{prop:elim}
Eliminating $\rho$ from $\mu$-reassigning each $b_i$ to $a_{i+1}$ and leaving all others unchanged-
yields another stable matching, denoted $\operatorname{elim}(\mu,\rho)$. Moreover, if
$\mu'=\operatorname{elim}(\mu,\rho)$ then:
\begin{itemize}
  \item for every $b\in B$, $\mu'(b)\succeq_b \mu(b)$, with strict improvement
        for $b\in\{b_1,\dots,b_\nu\}$ and $\mu'(b)=\mu(b)$ for $b\notin\{b_1,\dots,b_\nu\}$;
  \item symmetrically, for every $a\in A$, $\mu'(a)\preceq_a \mu(a)$, with strict worsening
        for $a\in\{a_1,\dots,a_\nu\}$ and $\mu'(a)=\mu(a)$ otherwise.
\end{itemize}
\end{proposition}

One can realize $\operatorname{elim}(\mu,\rho)$ by “breaking’’ one pair in $\rho$ and letting DA continue \emph{only} among the affected agents; the resulting local chain of proposals is exactly the rotation update.

\paragraph{Poset, encoding, and size.}
Define a partial order on rotations by
\[
\rho \preceq \rho' \quad\Longleftrightarrow\quad
\text{every elimination sequence that exposes $\rho'$ eliminates $\rho$ first},
\]
and write $\rho \prec \rho'$ for $\rho \preceq \rho'$ and $\rho\neq\rho'$.

\begin{proposition}[Rotation poset and down-set encoding]\label{prop:poset}
Every stable matching corresponds uniquely to a \emph{down-set} $D$ of the rotation poset (closed under predecessors), with
\[
\mu\ =\ \operatorname{elim}(\mu_A, D).
\]
There are $O(n^2)$ rotations in total, and any chain from $\mu_A$ to $\mu_B$ eliminates at most $O(n^2)$ rotations.
\end{proposition}

\begin{definition}[Rotation DAG / Hasse diagram]\label{def:precedence-dag}
Nodes are rotations. There is an arc $\rho\!\to\!\rho'$ iff $\rho\prec\rho'$ and there is no $\tilde\rho$ with
$\rho\prec\tilde\rho\prec\rho'$ (a covering edge). The transitive closure recovers the rotation poset.
This is exactly the "rotation DAG" referenced in the main text.
\end{definition}

In particular, $\operatorname{elim}(\mu_A,D)$ denotes eliminating all rotations in $D$ in any
linear extension of the precedence order; the result is well defined (independent of the chosen extension) because $D$ is a down-set of the rotation poset and any two incomparable rotations commute (in particular, rotations on disjoint sets of agents commute).

\begin{proposition}[Construction in $O(n^2)$]\label{prop:compute-rot}
From the DA \emph{shortlists}, the arcs that survive rejections during the DA run (i.e., edges consistent with the final tentative choices), one can construct all rotations and their precedence DAG in $O(n^2)$ time and space.
\end{proposition}

Think of rotations as nodes in a DAG; picking a down-set $D$ means choosing exactly which rotations to eliminate. Independent rotations commute, so their internal order does not matter.

\paragraph{Facts used in the paper.}
We repeatedly use: (i) $\mu_A \preceq \mu \preceq \mu_B$; (ii) additive objectives decompose over $D$ via $\mu=\operatorname{elim}(\mu_A,D)$; (iii) the $O(n^2)$ bound on rotations and exposure constraints underpins all $O(n^2)$-scale checks in the paper.

\section{Convex Program Formulations}\label{app:convex}

\subsection{Second-Order Cone Programming (SOCP)}\label{app:SOCP}
A \emph{second-order cone program (SOCP)} is a convex optimization problem of the form
\[
\begin{array}{ll}
\text{minimize}   & \mathbf{c}^\top \mathbf{x} \\[3pt]
\text{subject to} & \|\mathbf{A}_i \mathbf{x} + \mathbf{b}_i\|_2 \;\le\; \mathbf{c}_i^\top \mathbf{x} + d_i,\quad i=1,\dots,k,\\[3pt]
                  & \mathbf{F}\mathbf{x} = \mathbf{g}.
\end{array}
\]

Our $p=2$ feasibility checks fit this framework, since the constraint 
$\|\hat{\vecs}(b)-\vecs(b)\|_2\le r$
defines a second-order cone.

\subsection{LP/SOCP Formulations for Robustness Verification}\label{app:verify-formulations}

For completeness, we list the convex programs used in Section~\ref{sec:verify-given} 
to test whether a perturbation of raius~$r$ can make $(a,b)$ a blocking pair.  
Each instance enforces the admissible perturbation constraints for a given $(a,b,Q)$ with $|Q|\le k$.

\smallskip
\noindent\emph{$p=1$ (Manhattan distance) – LP feasibility}
\[
\begin{aligned}
\text{find }&\hat{\vecs}(b),\ \lambda>0,\ \mathbf{z} \in\mathbb{R}^m_{\ge0}\\
\text{s.t. }&
\sum_i \hat s_i(b)=1,\quad \hat s_i(b)\ge0,\\
&\hat s_i(b)=\lambda\, s_i(b)\ (i\notin Q),\\
&\hat{\vecs}(b)\cdot\boldsymbol{\Delta}(b;a\mid\mu)\le0,\\
&z_i\ge \hat s_i(b)-s_i(b),\quad 
z_i\ge s_i(b)-\hat s_i(b)\ (\forall i),\\
&\sum_i z_i\le r.
\end{aligned}
\]
\emph{Variables:} $\hat{\vecs}(b)\in\mathbb{R}^m_{\ge0}$, $\lambda>0$, $\mathbf{z}\in\mathbb{R}^m_{\ge0}$.

\smallskip
\noindent\emph{$p=2$ (Euclidean distance) – SOCP feasibility}
\[
\begin{aligned}
\text{find }&\hat{\vecs}(b),\ \lambda>0\\
\text{s.t. }&
\sum_i \hat s_i(b)=1,\quad \hat s_i(b)\ge0,\\
&\hat s_i(b)=\lambda\, s_i(b)\ (i\notin Q),\\
&\hat{\vecs}(b)\cdot\boldsymbol{\Delta}(b;a\mid\mu)\le0,\\
&\|\hat{\vecs}(b)-\vecs(b)\|_2\le r.
\end{aligned}
\]
\emph{Variables:} $\hat{\vecs}(b)\in\mathbb{R}^m_{\ge0}$, $\lambda>0$.

\smallskip
Both formulations test feasibility of a convex region defined by linear or second-order constraints.
If all instances are infeasible, the matching $\mu$ is $(k,r,p)$-robust; otherwise, a feasible instance identifies a blocking pair.

\subsection{Explicit LP/SOCP Formulations for the Maximum Robustness Radius}\label{app:radius-formulations}

This appendix provides the explicit convex programs used to compute 
$\rad^{\min}(b;a\mid Q)$ as defined in Definition~\ref{def:pairwise-thresholds}.
Each instance minimizes the perturbation radius~$r$ required to make $(a,b)$ a blocking pair, 
subject to the admissible perturbation constraints.
The formulation depends on the choice of the norm~$p$.
All variables and constraints follow the same conventions as in Section~\ref{sec:verify-given}.

For each pair $(a,b)$ and support $Q\subseteq[m]$ with $|Q|\le k$, 
the optimization problem takes one of the following forms.

\smallskip
\noindent\emph{$p=\infty$ (box distance) - LP for $(a,b,Q)$}
\[
\begin{aligned}
\min\ & r\\
\text{s.t. }&
\sum_i \hat s_i(b)=1,\quad \hat s_i(b)\ge0,\\
&\hat s_i(b)=\lambda\, s_i(b)\ (i\notin Q),\quad \lambda>0,\\
&\hat{\vecs}(b)\cdot\boldsymbol{\Delta}(b;a\mid\mu)\le0,\\
&-r\le \hat s_i(b)-s_i(b)\le r\quad(\forall i),\\
&r\ge0.
\end{aligned}
\]
\emph{Variables:} $\hat{\vecs}(b)\in\mathbb{R}^m_{\ge0}$, $\lambda>0$, $r\ge0$.

\bigskip
\noindent\emph{$p=1$ (Manhattan distance) - LP for $(a,b,Q)$}
\[
\begin{aligned}
\min\ & r\\
\text{s.t. }&
\sum_i \hat s_i(b)=1,\quad \hat s_i(b)\ge0,\\
&\hat s_i(b)=\lambda\, s_i(b)\ (i\notin Q),\quad \lambda>0,\\
&\hat{\vecs}(b)\cdot\boldsymbol{\Delta}(b;a\mid\mu)\le0,\\
&z_i\ge \hat s_i(b)-s_i(b),\quad 
z_i\ge s_i(b)-\hat s_i(b)\ (\forall i),\\
&\sum_i z_i\le r,\quad z_i\ge0\ (\forall i),\\
&r\ge0.
\end{aligned}
\]
\emph{Variables:} $\hat{\vecs}(b)\in\mathbb{R}^m_{\ge0}$, $\lambda>0$, $\mathbf{z}\in\mathbb{R}^m_{\ge0}$, $r\ge0$.

\bigskip
\noindent\emph{$p=2$ (Euclidean distance) - SOCP for $(a,b,Q)$}
\[
\begin{aligned}
\min\ & r\\
\text{s.t. }&
\sum_i \hat s_i(b)=1,\quad \hat s_i(b)\ge0,\\
&\hat s_i(b)=\lambda\, s_i(b)\ (i\notin Q),\quad \lambda>0,\\
&\hat{\vecs}(b)\cdot\boldsymbol{\Delta}(b;a\mid\mu)\le0,\\
&\|\hat{\vecs}(b)-\vecs(b)\|_2\le r,\\
&r\ge0.
\end{aligned}
\]
\emph{Variables:} $\hat{\vecs}(b)\in\mathbb{R}^m_{\ge0}$, $\lambda>0$, $r\ge0$.

\bigskip
\noindent
Since LP and SOCP infeasibility can be decided in polynomial time 
by interior-point methods~\cite{AlizadehGoldfarb03,BoydVandenberghe04}, 
and since $m$ and $k$ are constants, the overall runtime is polynomial in~$n$.

\section{Volume of the Robustness Region}\label{app:volume}

This appendix provides the full analysis of the volume computation for the robustness region~$\mathcal P_\mu$.
In general, computing the exact volume of a polytope defined by linear inequalities is \#P--hard~\cite{DyerFrieze88}.  
However, in our model $\mathcal P_\mu$ admits a product structure across the $B$–agents:
$\mathcal P_\mu = \prod_{b\in B} \mathcal P_\mu(b),$
where each factor $\mathcal P_\mu(b)$ is a rational polytope of affine dimension~$(m{-}1)$ inside~$\simplex{m}$,
bounded by $O(n)$ halfspaces.
This special structure enables both exact and approximate volume computation in polynomial time for fixed~$m$.

\begin{theorem}[Exact volume in polynomial time]\label{thm:volume-poly}
For any fixed $m$, the exact volume $\mathrm{Vol}(\mathcal P_\mu)$ can be computed in time polynomial in~$n$.
\end{theorem}
\begin{proof}

We show that the special product structure of the robustness region makes its volume exactly computable in polynomial time.
We proceed in three steps.

\medskip
\noindent
\textbf{Step 1: Structure of each factor.}
From Lemma~\ref{lem:factorization}, each factor 
$\mathcal P_\mu(b)$ is the intersection of the simplex~$\simplex{m}$
with at most $O(n)$ halfspaces defined by rational coefficients 
(from the input attribute data).
Hence $\mathcal P_\mu(b)$ is a rational polytope of constant affine dimension $(m-1)$ 
embedded in~$\mathbb{R}^m$.
Since the dimension is constant, its description complexity grows only linearly with~$n$. This linear bound guarantees that the volume of each factor 
can be computed in polynomial time in~$n$.

\medskip
\noindent
\textbf{Step 2: Exact volume of a single factor.}
Each $\mathcal P_\mu(b)$ is a rational polytope of fixed affine dimension $(m-1)$
defined by $O(n)$ linear inequalities.
In fixed dimension, the exact volume of such a polytope
can be computed in $n^{O(m)}$ time using Lawrence’s cone decomposition~\cite{Lawrence91}.
Applying this to each factor yields $\mathrm{Vol}(\mathcal P_\mu(b))$ exactly.

\medskip
\noindent
\textbf{Step 3: Factorization across the $B$–agents.}
By Lemma~\ref{lem:factorization},
the robustness region decomposes as
\[
\mathcal P_\mu \;=\; \prod_{b\in B}\mathcal P_\mu(b),
\qquad
\text{hence}\qquad
\mathrm{Vol}(\mathcal P_\mu) \;=\; \prod_{b\in B} \mathrm{Vol}(\mathcal P_\mu(b)).
\]
Since each local volume $\mathrm{Vol}(\mathcal P_\mu(b))$ can be computed in time $n^{O(m)}$,
and the product involves $|B| = O(n)$ factors,
the total computation of $\mathrm{Vol}(\mathcal P_\mu)$ requires $n^{O(m)}$ time overall.
All arithmetic operations are exact over the rationals.
\end{proof}

\paragraph{Approximate volume for moderate dimensions.}
While the exact algorithm is polynomial for fixed~$m$, 
its running time $n^{O(m)}$ grows rapidly and becomes impractical 
even for moderate dimensions (e.g., $m>5$).  
In such cases, the factorized polyhedral structure of $\mathcal P_\mu$ 
admits efficient randomized approximation.  
Standard hit-and-run sampling methods for convex bodies~\cite{LovaszVempala06HR} 
yield a fully polynomial randomized approximation scheme (FPRAS) 
that estimates $\mathrm{Vol}(\mathcal P_\mu)$ within relative error~$\varepsilon$ 
in time $\tilde{O}((nm)^5/\varepsilon^2)$, 
where $\tilde{O}$ hides polylogarithmic factors.

The robustness region thus supports both exact and approximate volume computation:  
exact volume in $n^{O(m)}$ time (practical for small~$m$), 
and approximation via an FPRAS 
in $\tilde{O}((nm)^5/\varepsilon^2)$ time for higher~$m$.

\end{document}